\let\accentvec\vec  
\documentclass{llncs}
\let\vec\accentvec 

\usepackage{comment}
\usepackage{booktabs}
\usepackage{url}
\usepackage{graphicx} 	
\usepackage[tight]{subfigure} 
\usepackage{float}

\usepackage{amsmath}
\usepackage{amssymb}
\usepackage{amsfonts}
\usepackage{amstext}
\usepackage{xspace}
\usepackage{graphicx}
\usepackage{mathrsfs}
\usepackage[shortlabels]{enumitem}
\setlist{noitemsep}

\newcommand{\yes}{\textsc{yes}\xspace}
\newcommand{\no}{\textsc{no}\xspace}

\newtheorem{numberedclaim}{Claim}
\newtheorem{observation}{Observation}

\newcommand{\claimqed}{\renewcommand{\squareforqed}{$\lrcorner$}\qed\renewcommand{\squareforqed}{\plainsquareforqed}}

\newcommand{\Oh}{{\mathcal{O}}}

\newcommand{\C}{{\mathcal{C}}}

\let\SSign\S
\renewcommand{\S}{\ensuremath{\mathcal{S}}\xspace}
\renewcommand{\P}{\ensuremath{\mathcal{P}}\xspace}

\newcommand{\opt}{\ensuremath{\mathrm{\textsc{opt}}}\xspace}

\newcommand{\containment}[0]{{\sf NP}~$\subseteq$~{\sf coNP$/$poly}\xspace}

\newcommand{\h}[1]{\end{document}}

\usepackage{boxedminipage}

\renewcommand{\C}{\ensuremath{\mathcal{C}}\xspace}
\newcommand{\F}{\ensuremath{\mathcal{F}}\xspace}

\newcommand{\defparproblem}[4]{
 \vspace{1mm}
\noindent\fbox{
 \begin{minipage}{0.96\textwidth}
 \begin{tabular*}{\textwidth}{@{\extracolsep{\fill}}lr} #1 & {\bf{Parameter:}} #3 \\ \end{tabular*}
 {\bf{Input:}} #2 \\
 {\bf{Question:}} #4
 \end{minipage}
 }
 \vspace{1mm}
}

\newcommand{\defproblem}[3]{
 \vspace{1mm}
\noindent\fbox{
 \begin{minipage}{0.96\textwidth}
 \begin{tabular*}{\textwidth}{@{\extracolsep{\fill}}lr} #1 &  \\ \end{tabular*}
 {\bf{Input:}} #2 \\
 {\bf{Question:}} #3
 \end{minipage}
 }
 \vspace{1mm}
}

\newcommand{\HitPathsInFlower}{\textsc{Hitting Paths in a Flower with Budgets}\xspace}
\newcommand{\HitPathsInGraph}{\textsc{Hitting Paths in a Graph}\xspace}
\newcommand{\SignedThreeSatn}{\textsc{$n$-Totally Ordered Regular Signed 3-SAT}\xspace}

\newcommand{\nTORSThreeSat}{\textsc{$n$-TORS 3-SAT}\xspace}

\newcommand{\kClique}{\textsc{$k$-Clique}\xspace}
\newcommand{\TORSTwoSat}{\textsc{TORS 2-SAT}\xspace}

\newcommand{\VertexCover}{\textsc{Vertex Cover}\xspace}
\newcommand{\HittingSet}{\textsc{Hitting Set}\xspace}
\newcommand{\SetCover}{\textsc{Set Cover}\xspace}

\title{On Structural Parameterizations of Hitting Set: \\ Hitting Paths in Graphs Using 2-SAT\thanks{Supported by NWO Veni grant ``Frontiers in Parameterized Preprocessing'' and NWO Gravity grant ``Networks''.}}

\author{Bart M.\ P.\ Jansen \inst{1}}

\institute{Eindhoven University of Technology, The Netherlands. \email{b.m.p.jansen@tue.nl}}

\usepackage[colorlinks=true,urlcolor=red,citecolor=red]{hyperref}

\graphicspath{{./images/}}

\begin{document}

\hypersetup{bookmarksdepth=-1}

\maketitle

\hypersetup{bookmarksdepth=2} 

\begin{abstract}
\HittingSet is a classic problem in combinatorial optimization. Its input consists of a set system~$\F$ over a finite universe~$U$ and an integer~$t$; the question is whether there is a set of~$t$ elements that intersects every set in~$\F$. The \HittingSet problem parameterized by the size of the solution is a well-known W[2]-complete problem in parameterized complexity theory. In this paper we investigate the complexity of \HittingSet under various structural parameterizations of the input. Our starting point is the folklore result that \HittingSet is polynomial-time solvable if there is a tree~$T$ on vertex set~$U$ such that the sets in~$\F$ induce connected subtrees of~$T$. We consider the case that there is a treelike graph with vertex set~$U$ such that the sets in~$\F$ induce connected subgraphs; the parameter of the problem is a measure of how treelike the graph is. Our main positive result is an algorithm that, given a graph~$G$ with cyclomatic number~$k$, a collection~$\P$ of simple paths in~$G$, and an integer~$t$, determines in time~$2^{5k} (|G| +|\P|)^{\Oh(1)}$ whether there is a vertex set of size~$t$ that hits all paths in~$\P$. It is based on a connection to the 2-SAT problem in multiple valued logic. For other parameterizations we derive W[1]-hardness and para-NP-completeness results.
\end{abstract}

\section{Introduction}

\HittingSet is a classic problem in combinatorial optimization that asks, given a set system~$\F$ over a finite universe~$U$, and an integer~$t$, whether there is a set of~$t$ elements that intersects every set in~$\F$. It was one of the first problems to be identified as NP-complete~\cite{Karp72}. Parameterized complexity theory is a refined view of computational complexity that aims to attack NP-hard problems by algorithms whose running time is exponential in a problem-specific \emph{parameter value}, but polynomial in terms of the overall input size. The standard parameterization of \HittingSet by the size of the desired solution is unlikely to admit such a fixed-parameter tractable algorithm, as it is W[2]-complete~\cite{DowneyF13}. The goal of this paper is to consider other parameterizations of \HittingSet, with the aim of obtaining FPT algorithms. Our starting point is the folklore result that \HittingSet is polynomial-time solvable when there is a tree~$T$ on vertex set~$U$ such that all sets~$S \in \F$ induce connected subtrees of~$T$. The \HittingSet problem on such an instance can be solved by a greedy strategy (Section~\ref{section:prelims}). Motivated by this result, we consider whether \HittingSet can be solved efficiently if there is a graph~$G$ that is close to being a tree, such that all~$S \in \F$ induce connected subgraphs of~$G$. We therefore parameterize the problem by measures of closeness of~$G$ to a tree, which forms an example of parameterizing by distance from triviality~\cite{Niedermeier10}.

\textbf{Our results.} 
One way to measure how close a connected graph is to a tree is to consider its \emph{cyclomatic number}~$k := m - (n - 1)$. This is the size of a minimum feedback edge set of the graph, i.e., of a minimum set of edges whose removal breaks all cycles in the graph. As a tree has cyclomatic number zero, it is natural to ask if \HittingSet can be solved efficiently if the set system~$\F$ can be represented by a graph~$G$ on vertex set~$U$ having small cyclomatic number, such that every set~$S \in \F$ induces a connected subgraph of~$G$. To decouple the difficulty of finding a representation of~$\F$ in this form from the problem of exploiting this representation to solve \HittingSet, we consider the situation when such a representation is given. In this setting, the problem can be phrased more naturally in graph-theoretical terms: given a graph~$G$ of cyclomatic number~$k$, a collection~$\S$ of connected subgraphs of~$G$, and an integer~$t$, is there a vertex set of size~$t$ that hits all subgraphs in~$\S$? 

\begin{table}[t]
	\caption{Parameterized complexity overview for hitting subgraphs by the minimum number of vertices, parameterized by measures of structure of the host graph. 
	}
	\centering
{
\begin{tabular}{@{}llllll@{}}
\toprule
parameter & \multicolumn{5}{c}{complexity for type of subgraphs to be hit} \\ 
\cmidrule{2-6}
& \multicolumn{2}{l}{path} & \phantom{abc} & \multicolumn{2}{l}{3-leaf subtree} \\ 
\midrule
cyclomatic number & FPT, no~$k^{\Oh(1)}$ kernel & thm.~\ref{theorem:pathsingraph:fpt} & & W[1]-hard & thm.~\ref{theorem:hitclaws:whard} \\
feedback vertex number & para-NP-complete & thm.~\ref{theorem:hitpaths:fvs:npc} & & para-NP-complete & thm.~\ref{theorem:hitpaths:fvs:npc} \\
\bottomrule
\end{tabular}
\label{table:summary}
}
\end{table}

Our first result for the parameterization by cyclomatic number is a hardness proof showing this problem to be W[1]-hard. In fact, we prove W[1]-hardness even when all subgraphs in~$\S$ are trees with at most three leaves. To establish this hardness result we prove that a variation of 3-SAT in multiple valued logic (see Section~\ref{section:prelims}) is W[1]-hard, which may be of independent interest. Concretely, we show the following. Given a set of~$n$ variables~$x_1, \ldots, x_n$ that can take values from~$1$ to~$N$, and a formula that is a conjunction of clauses of size at most three, where each literal is of the form~$x_i \geq c$ or~$x_i \leq c$ for~$c \in [N]$, it is W[1]-hard parameterized by~$n$ to determine whether there is an assignment to the variables satisfying all clauses. This parameterized logic problem reduces to the discussed structural parameterization of \HittingSet in a natural way.

The hardness result motivates us to place further restrictions on the problem in search of fixed-parameter tractable cases. We consider the situation of hitting a set~$\P$ of \emph{simple paths} in a graph~$G$ of cyclomatic number~$k$. This corresponds to \HittingSet instances where there is a graph~$G$ on~$U$ such that for all sets~$S$ in~$\F$, there is a \emph{simple path} in~$G$ on vertex set~$S$. We prove that this problem is fixed-parameter tractable and can be solved in time~$2^{5k} (|G| + |\P|)^{\Oh(1)}$, which is the main algorithmic result in this paper. The algorithm is based on a reduction to~$2^{5k}$ instances of the 2-SAT problem in multiple valued logic, which is known to be polynomial-time solvable~\cite{BejarHM01,Manya00}. The reduction exploits the fact that in tree-like parts of the graph, the local structure of minimum hitting sets can be determined by greedily computed optimal hitting sets for subtrees of a tree. After branching in~$2^{5k}$ directions to determine the form of a solution, the interaction between such canonical subsolutions is then encoded in a 2-SAT formula in multiple valued logic, which can be evaluated efficiently.

There are several other parameters that measure the closeness of a graph to a tree, such as the \emph{feedback vertex number} and \emph{treewidth} (cf.~\cite{FellowsJR13}). As these parameters have smaller values than the cyclomatic number, one might hope to extend the FPT result mentioned above to these parameters. However, we show that this is impossible, unless P=NP. In particular, we prove that the problem of hitting simple paths in a graph of feedback vertex number~$2$ is NP-complete, showing the parameterizations by feedback vertex number and treewidth to be para-NP-complete. Table~\ref{table:summary} gives an overview of the results in this paper.

\textbf{Related work.} Several authors~\cite{CoppersmithV85,Fiala01,UhlmannW13} have considered problems parameterized by cyclomatic number; this is also known as parameterizing by feedback edge set. 
In parameterized complexity, \HittingSet is often studied when the sets to be hit have constant size. In this setting, several FPT algorithms and kernelizations bounds are known \cite{Abu-Khzam10,DellM14,Wahlstrom07}. The weighted \SetCover problem, which is dual to \HittingSet, has been analyzed for tree-like set systems by Guo and Niedermeier~\cite{GuoN06}. Recently, Lu et al.~\cite{LuLTLX14} considered \SetCover and \HittingSet for set systems representable as subtrees of a (restricted type of) tree, distinguishing polynomial-time and NP-complete cases.

\textbf{Organization.} Preliminaries are given in~\ref{section:prelims}. The FPT algorithm for hitting paths is developed in Section~\ref{section:fpt}. Section~\ref{section:hardness} contains the hardness proofs.

\section{Preliminaries} \label{section:prelims}

\textbf{Parameterized complexity.}
A parameterized problem is a set~$Q \subseteq \Sigma^* \times \mathbb{N}$, where~$\Sigma$ is a fixed finite alphabet. The second component of a tuple~$(x,k) \in \Sigma^* \times \mathbb{N}$ is the \emph{parameter}. A parameterized problem is (strongly uniformly) \emph{fixed-parameter tractable} if there is an algorithm that decides every input~$(x,k)$ in time~$f(k)|x|^{\Oh(1)}$. Evidence that a problem is not fixed-parameter tractable is given by proving that it is W[1]-hard. We refer to one of the textbooks~\cite{DowneyF13,FlumG06} for more background.

\textbf{Graphs.}
All graphs we consider are simple, undirected and finite. A graph~$G$ consists of a set of vertices~$V(G)$ and edges~$E(G)$. Notation not defined here is standard. For a set of vertices~$S$ we denote by~$N_G(S)$ the set~$\bigcup _{v \in S} N_G(v) \setminus S$. A path in a graph~$G$ is a sequence of distinct vertices such that successive vertices are connected by an edge. The first and last vertices on the path are its endpoints, the remaining vertices are its interior vertices. Given a graph~$G$ and a vertex subset~$S \subseteq V(G)$, the operation of \emph{identifying} the vertices of~$S$ into a new vertex~$z$ is performed as follows: delete the vertices in~$S$ and their incident edges, and insert a new vertex~$z$ that is adjacent to~$N_G(S)$, i.e., to all remaining vertices of~$G$ that were adjacent to at least one member of~$S$.

\begin{proposition} \label{proposition:degtwo}
Let~$G$ be a connected graph of minimum degree at least two with cyclomatic number~$k$. The number of vertices in~$G$ with degree at least three is bounded by~$2k-2$.
\end{proposition}
\begin{proof}
Denote by~$n_2$ and~$n_{\geq 3}$ the number of vertices in~$G$ with degree two and at least three, respectively. Let~$n$ and~$m$ be the total number of vertices and edges in~$G$, and let~$d(v)$ denote the degree of a vertex~$v$. Since~$k = m - (n-1)$ we have~$m = k + (n_2 + n_{\geq 3} - 1)$. The value of~$m$ can also be obtained as half the degree sum of~$G$:
$$m = \frac{1}{2} \sum_{v \in V(G)} d(v) \geq n_2 + \frac{3n_{\geq 3}}{2}.$$
Hence we find:
$$m = k + (n_2 + n_{\geq 3} - 1) \geq n_2 + \frac{3n_{\geq 3}}{2},$$ from which we obtain~$n_{\geq 3} \leq 2k-2$ by subtracting~$n_2 + n_{\geq 3}$ on both sides and multiplying by two.
\qed
\end{proof}

\begin{proposition} \label{proposition:numcomponents}
Let~$G$ be a connected graph of minimum degree at least two with cyclomatic number~$k$ and let~$S$ be the set of vertices of degree at least three. If~$S \neq \emptyset$ then the number of connected components of~$G - S$ is at most~$k + |S| - 1$.
\end{proposition}
\begin{proof}
As~$S$ contains all vertices of degree at least three, every connected component~$C$ of~$G - S$ is a path. Since~$G$ has minimum degree at least two, every endpoint of such a path has a neighbor in~$S$. Hence for every connected component~$C$ of~$G - S$ there are exactly two edges between~$C$ and~$S$. Consider the multigraph~$H$ on vertex set~$S$ defined as follows. For every component~$C$ of~$G - S$, consider the two edges between~$C$ and~$S$  and let~$x,y$ be their endpoints in~$S$. We add an edge between~$x$ and~$y$ to~$H$; if~$x=y$ this becomes a self-loop, and there is the chance of creating parallel edges. Since~$H$ is a connected topological minor of~$G$ it is easy to see that the cyclomatic number~$k'$ of~$H$ does not exceed that of~$G$. Since~$|E(H)| = (|V(H)| - 1) + k'$, we find that~$|E(H)| = |S| - 1 + k' \leq k + |S| - 1$. Since connected components of~$G - S$ are in 1-to-1 correspondence with edges of~$H$, this completes the proof.
\qed
\end{proof}

\textbf{Hitting set.} 
A set system~$\F \subseteq 2^U$ can be viewed as a hypergraph whose vertices are~$U$ and whose hyperedges are formed by the sets in~$\F$. A set system~$\F$ is a \emph{hypertree} if there is a tree~$T$ on vertex set~$U$ such that every set in~$\F$ induces a subtree of~$T$. Testing whether a set system is a hypertree, and constructing a tree representation if this is the case, can be done in polynomial time~\cite{Trick87}. 

We frequently use the fact that a minimum hitting set for a hypertree can be found in polynomial time (cf.~\cite[\SSign 2]{GuoN06} for a view from a dual perspective). When a tree representation is known, a greedy algorithm can be used to find a minimum hitting set. If we root the tree at a leaf and find a vertex~$v$ of maximum depth for which there is a set~$S \in \F$ whose members all belong to the subtree rooted at~$v$, then it is easy to show there is a minimum hitting set containing~$v$. Consequently, we may add~$v$ to the solution under construction, remove all sets hit by~$v$, and remove all elements in the subtree rooted at~$v$ from the universe.

This idea can be extended for the following setting. Suppose we have a graph~$G$ that is isomorphic to a simple cycle and a set~$\P$ of paths in~$G$. To find a minimum vertex set that hits all the paths in~$\P$, we try for each vertex~$v$ of~$G$ whether there is a minimum solution containing it. After removing~$v$ and the paths hit by~$v$, the remaining structure is a hypertree since the cycle breaks open when removing~$v$. The minimum over all choices of~$v$ gives an optimal hitting set. We will use this in our FPT algorithm to deal with a corner case.

\textbf{Multiple valued logic.} 
The hitting set problems we are interested in turn out to be related to variations of the \textsc{Satisfiability} problem that have been studied in the field of multiple valued logic. In a multiple valued logic, variables can take on more values than just~$0$ and~$1$: there is a \emph{truth value set} containing the possible values. For our application, the truth value set is totally ordered; it is a range of integers~$[N] = \{1, \ldots, N\}$. A \emph{regular sign} is a constraint of the form~$\geq j$ or~$\leq j$ for~$j \in [N]$. By constraining variables with regular signs, resulting in (generalized) literals of the form~$x_i \geq j$ or~$x_i \leq j$, and combining such literals with the usual logical connectives, one creates totally ordered regular signed formulas. As expected, the satisfiability problem for such formulas is to determine whether every variable can be assigned a value in the range~$[N]$ such that the formula is satisfied. We shall be interested in the case of CNF formulas with clauses having at most two (2-SAT) or at most three (3-SAT) literals. 

\defparproblem{\SignedThreeSatn}
{A totally ordered regular signed 3-CNF formula with~$n$ variables and truth value set~$[N]$.}
{$n$.}
{Is the formula satisfiable?}

For brevity we sometimes refer to this problem as \nTORSThreeSat. We also consider \TORSTwoSat, where clauses have at most two literals, which is polynomial-time solvable~\cite{Manya00}. In particular, \TORSTwoSat can be reduced to the 2-SAT problem in classical logic~\cite[\SSign 3]{BejarHM01}, which is well-known to be solvable in linear time~\cite{AspvallPT79}. For completeness, we sketch the reduction in Appendix~\ref{app:twosat}.

\section{Algorithms} \label{section:fpt}

The goal of this section is to develop an FPT algorithm for the following parameterized problem.

\defparproblem{\HitPathsInGraph}
{An undirected simple graph~$G$ with cyclomatic number~$k$, an integer~$t$, and a set~$\P$ of simple paths in~$G$.}
{$k$.}
{Is there a set~$X \subseteq V(G)$ of size at most~$t$ that hits all paths in~$\P$?}

The algorithm consists of two reductions. An instance of \HitPathsInGraph is reduced to a hitting set problem on a more structured graph, called a flower. An instance with such a flower structure can be reduced to a polynomial-time solvable 2-SAT problem in multiple valued logic. This section is structured as follows. We first describe the flower structure and the reduction to 2-SAT in Section~\ref{subsection:hit:flowers}. Afterward we show how to build an FPT algorithm from this ingredient, in Section~\ref{subsection:hit:paths}.

\subsection{Hitting Paths in Flowers} \label{subsection:hit:flowers}

The key notion in this section is that of a \emph{flower graph}, which is a graph~$G$ with a distinguished vertex~$z$ called the \emph{core} such that all connected components of~$G-\{z\}$ are paths~$R_1, \ldots, R_n$ of which no interior vertex is adjacent to~$z$. These paths are called \emph{petals} of the flower. When working with flower graphs we will assume an arbitrary but fixed ordering of the petals as~$R_1, \ldots, R_n$, as well as an orientation of each petal~$R_i$ as consisting of vertices~$r_{i,1}, \ldots, r_{i,|V(R_i)|}$. For ease of discussion we will interpret each petal to be laid out from left to right in order of increasing indices. We will give an FPT branching algorithm that reduces \HitPathsInGraph to solving several instances of the following more restricted problem.

\defproblem{\HitPathsInFlower}
{A flower graph~$G$ with core~$z$ and petals~$R_1, \ldots, R_n$, a set of simple paths~$\P = \{P_1, \ldots, P_m\}$ in~$G$, and a budget function~$b \colon [n] \to \mathbb{N}_{\geq 1}$.}
{Is there a set~$X \subseteq V(G) \setminus \{z\}$ that hits all paths in~$\P$ such that~$|X \cap V(R_i)| = b(i)$ for all~$i \in [n]$?}

We show that \HitPathsInFlower can be solved in polynomial time. The following notion will be instrumental to analyze the structure of solutions to this problem.

\begin{definition} \label{definition:canonical}
Let~$R_i$ be a petal of an instance~$(G,z,\P,b)$ of \HitPathsInFlower and let~$1 \leq \ell \leq |V(R_i)|$. The \emph{canonical $\ell$-th solution} for petal~$R_i$ is defined by the following process.
\begin{enumerate}
	\item If there is a path in~$\P$ that is contained entirely within~$\{r_{i,1}, \ldots, r_{i,\ell-1}\}$, then define the canonical $\ell$-th solution to be NIL.\label{step:earlyout}
	\item Otherwise, initialize~$X_{i,\ell}$ as the singleton set containing~$r_{i,\ell}$.
	\begin{enumerate}
		\item While there is a path in~$\P$ that is contained entirely within~$R_i$ and is not intersected by~$X_{i,\ell}$, consider a path among this set that minimizes the index~$j'$ of its right endpoint and add~$r_{i,j'}$ to~$X_{i,\ell}$.\label{step:hitpath}
		\item While~$|X_{i,\ell}| < b(i)$ and~$Y := \{r_{i,\ell}, \ldots, r_{i,|V(R_i)|}\} \setminus X_{i,\ell} \neq \emptyset$, add the highest-indexed vertex from~$Y$ to~$X_{i,\ell}$. (Recall that~$b(i)$ is the budget for petal~$R_i$.)\label{step:fillsize}
		\item If~$|X_{i,\ell}| = b(i)$, the canonical $\ell$-th solution is~$X_{i,\ell}$. If~$|X_{i,\ell}| \neq b(i)$, define the canonical $\ell$-th solution to be NIL.\label{step:toolarge}
		\end{enumerate}
\end{enumerate}
A set~$X_i \subseteq V(R_i)$ is a canonical solution for petal~$R_i$ if there is an integer~$\ell$ for which~$X_i$ is the canonical $\ell$-th solution for~$R_i$. A canonical solution is well defined if it is not NIL. A solution~$X$ to the instance~$(G,z,\P,b)$ is \emph{globally canonical} if~$X \cap R_i$ is a well-defined canonical solution for all~$i$.
\end{definition}

Figure~\ref{fig:petals} illustrates these concepts. For a set~$X_i \subseteq V(R_i)$ we will denote by~$\max(X_i)$ the highest index of any vertex in~$X_i$, i.e., the index of the rightmost vertex of~$X_i$. Similarly, we denote by~$\min(X_i)$ the index of the leftmost vertex of~$X_i$. The following observations about the procedure will be useful.

\begin{figure}[t]
\begin{center}
\subfigure[Flower.]{\label{fig:petal0}
\includegraphics[scale=1]{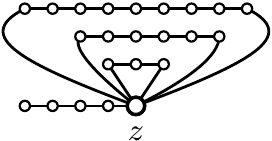}
}
\subfigure[$9$-vertex petal.]{\label{fig:petal1}
\includegraphics[scale=1]{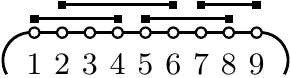}
}
\subfigure[Canonical solution.]{\label{fig:petal2}
\includegraphics[scale=1]{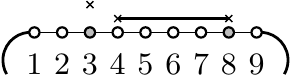}
}
\caption{(\ref{fig:petal0}) A flower graph with~$4$ petals and core~$z$. (\ref{fig:petal1}) A $9$-vertex petal whose endpoints are adjacent to~$z$. The target paths within the petal that must be hit by a solution are drawn stacked on top of each other. (\ref{fig:petal2}) The set~$\{3,8\}$ is the $3$-rd canonical solution of size~$2$ for the petal, with respect to the target paths drawn in \ref{fig:petal1}. The corresponding partition of~$\{3,\ldots,8\}$ into two subpaths described in Observation~\ref{observation:partition} is shown above the petal. It includes the singleton path~$\{3\}$. The canonical $1$-st solution of size~$2$ is NIL, since the procedure of Definition~\ref{definition:canonical} produces the set~$X_{i,\ell} = \{1,6,9\}$, which is too large and is rejected in Step~\ref{step:toolarge}.}
\end{center}\label{fig:petals}
\end{figure}

\begin{observation} \label{observation:leftmost}
If~$X_{i,\ell}$ is a well-defined canonical solution, then~$\min(X_{i,\ell}) = \ell$.
\end{observation}

\begin{observation} \label{observation:partition}
Let~$X_{i,\ell}$ result from Definition~\ref{definition:canonical}, and assume that Step~\ref{step:earlyout} does not apply and that Step~\ref{step:fillsize} is never triggered during the procedure. Partition the interval~$\{r_{i,\ell}, \ldots, r_{i,\max(X_{i,\ell})}\}$ into~$|X_{i,\ell}|$ maximal subpaths that each end at a vertex of~$X_{i,\ell}$ and contain no other vertices of~$X_{i,\ell}$. Then, for every such subpath~$R'$ except the singleton subpath~$\{r_{i,\ell}\}$, there is a path in~$\P$ contained entirely within~$R'$.
\end{observation}

The main strategy behind our reduction of \HitPathsInFlower to \TORSTwoSat will be as follows. We will show that, if a solution to the hitting set problem exists, then there is a globally canonical solution. Such a solution can be fully characterized by indicating, for each petal, the index of the canonical solution on the petal (i.e., the leftmost vertex of the petal that is in the solution). Hence finding a solution reduces to finding a choice of canonical solutions on the petals. It turns out that for every path~$P \in \P$, one can create a signed 2-clause on the variables controlling the choices on two petals, such that the path is hit by the selected solution if and only if the indices of the canonical subsolutions satisfy the 2-clause. This allows the hitting set problem to be modeled by \TORSTwoSat. We now formalize these ideas. Let us first get a feeling for canonical solutions by proving the following lemma.

\begin{lemma} \label{lemma:contiguous}
Let~$(G,z,\P,b)$ be an instance of \HitPathsInFlower and let~$R_i$ be a petal. The indices for which~$R_i$ has a well-defined canonical solution form a contiguous set of integers.
\end{lemma}
\begin{proof}
Assume for a contradiction that there are~$\ell_1 < \ell_2 < \ell_3$ such that the canonical solutions for~$\ell_1$ and~$\ell_3$ are well-defined, but that for~$\ell_2$ is not. Let us consider why the canonical solution for~$\ell_2$ is not well-defined. 
\begin{enumerate}
	\item If Step~\ref{step:earlyout} applies for~$\ell_2$, then the path~$P \in \P$ that is contained entirely within~$\{r_{i,1}, \ldots, r_{i,\ell_2}\}$ also causes Step~\ref{step:earlyout} to apply for~$\ell_3 > \ell_2$, contradicting the assumption that there is a well-defined canonical solution for~$\ell_3$.
	\item If~$X_{i,\ell_2}$ is too small in Step~\ref{step:toolarge}, then this implies that there are less than~$b(i)$ vertices in~$\{r_{\ell_2}, \ldots, r_{\ell_{|V(R_i)|}}\}$. But this contradicts the fact that~$X_{i,\ell_3}$ has~$b(i)$ vertices and is a subset of~$\{r_{i,\ell_3}, \ldots, r_{i,|V(R_i)|}\}$ for~$\ell_3 > \ell_2$.
	\item If~$X_{i,\ell_2}$ is too large in Step~\ref{step:toolarge}, then its size exceeds~$|X_{i,\ell_1}| = b(i)$. Hence the precondition to Step~\ref{step:fillsize} never applied during the procedure for~$\ell_2$. Consider the partition of the interval~$\{r_{i,\ell_2}, \ldots, r_{i,\max (X_{i,\ell_2})}\}$ into~$|X_{i,\ell_2}|$ subpaths as described in Observation~\ref{observation:partition}. Every subpath in the partition contains exactly one vertex of~$X_{i,\ell_2}$, and all vertices of~$X_{i,\ell_2}$ are in one such subpath. Observe that~$r_{i,\ell_1} \in X_{i,\ell_1} \setminus X_{i,\ell_2}$, such that~$X_{i,\ell_1}$ contains at most~$b(i) - 1$ vertices in the interval~$r_{i,\ell_2}, \ldots, r_{i,|V(R_i)|}$. Since~$|X_{i,\ell_2}| > |X_{i,\ell_1}| = b(i)$, there are at least two subpaths in the partition from which~$X_{i,\ell_1}$ contains no vertex. Hence there is such a subpath, say~$R' := \{r_{i,p}, \ldots, r_{i,q}\}$, that is not the singleton path~$\{r_{i,\ell_2}\}$ and that contains no vertices of~$X_{i,\ell_1}$. Then, by Observation~\ref{observation:partition}, there is a target path~$P$ in~$\P$ that is entirely contained within~$R'$. But~$X_{i,\ell_1}$ contains no vertex of this path, showing that~$X_{i,\ell_1}$ does not intersect~$P$, which contradicts the fact that the while-loop of Step~\ref{step:hitpath} terminated when defining~$X_{i,\ell_1}$.
\end{enumerate}
As the cases are exhaustive, this concludes the proof.
\qed
\end{proof}

As the procedure of Definition~\ref{definition:canonical} can be implemented in polynomial time, the set of indices for which a petal has a canonical solution can be computed in polynomial time. We continue describing the structure of canonical solutions.

\begin{lemma} \label{lemma:latersolutionsreachfurther}
Let~$(G,z,\P,b)$ be an instance of \HitPathsInFlower and let~$R_i$ be a petal. If~$\ell_1 < \ell_2$, and the $\ell_1$-th and the~$\ell_2$-th canonical solutions are well-defined as~$X_{i,\ell_1}$ and~$X_{i,\ell_2}$, then~$\max(X_{i,\ell_1}) \leq \max(X_{i,\ell_2})$.
\end{lemma}
\begin{proof}
Assume that~$\max(X_{i,\ell_1}) > \max(X_{i,\ell_2})$. We aim to apply Observation~\ref{observation:partition} to derive a contradiction. Since both canonical solutions are well defined, Step~\ref{step:earlyout} does not apply to~$\ell_2$. As our assumption implies that the rightmost vertex of~$R_i$ is not in~$X_{i,\ell_2}$, it follows that Step~\ref{step:fillsize} never applied during the procedure for~$\ell_2$. Since at least one vertex of~$X_{i,\ell_1}$ lies right of~$\max(X_{i,\ell_2})$, and~$X_{i,\ell_1}$ contains vertex~$r_{i,\ell_1}$ that lies left of~$r_{i,\ell_2}$, the partition of~$\{r_{i,\ell_2}, \ldots, r_{i,\max(X_{i,\ell_2})}\}$ into~$b(i)$ subpaths described by Observation~\ref{observation:partition} contains at least two subpaths from which~$X_{i,\ell_1}$ contains no vertex. Hence there is such a subpath~$R'$ that is not intersected by~$X_{i,\ell_1}$ for which there is a target path~$P \in \P$ contained entirely within~$R'$. This contradicts the fact that~$X_{i,\ell_1}$ hits all paths contained entirely within~$R_i$ by Step~\ref{step:hitpath} of Definition~\ref{definition:canonical}.
\qed
\end{proof}

We now establish that the hitting set problem has a globally canonical solution, if it has a solution at all. The proof exploits the fact that, after selecting the leftmost vertex of a petal to be used in the hitting set, removing it from the graph, and removing the paths hit by this vertex from the graph, the remainder of the petal turns into a pendant path that connects to the rest of the graph at vertex~$z$. The hitting set problem has a greedy solution within this resulting path, which reflects the structure of the canonical solution. Formalizing this line of reasoning is tedious but straight-forward.

\begin{lemma} \label{lemma:existscanonical}
Let~$(G,z,\P,b)$ be an instance of \HitPathsInFlower having petals~$R_1, \ldots, R_n$. If the instance has a solution~$X'$, then it has a globally canonical solution~$X$. 
\end{lemma}
\begin{proof}
Proof by induction on the number~$k$ of petals for which~$X' \cap R_i$ is not a canonical solution. When~$k=0$ the claim is trivial, so assume~$k > 0$ and let~$R_i$ be a petal such that~$X' \cap V(R_i)$ is not a canonical solution. Since~$X'$ is a solution, by the definition of \HitPathsInFlower we have~$|X' \cap V(R_i)| = b(i)$. Let~$\ell$ be the index of the leftmost vertex from~$X'$ on~$R_i$. Since~$b(i) \geq 1$, such a vertex exists. Let~$X'_{i,\ell} := X' \cap V(R_i)$.

\begin{claim}
The $\ell$-th canonical solution for~$R_i$ is well defined.
\end{claim}
\begin{proof}
Consider the set~$X_{i,\ell}$ resulting from the process of Definition~\ref{definition:canonical} and assume for a contradiction that the process defines the canonical $\ell$-th solution to be NIL. There are two cases that yield NIL; we treat them consecutively.
\begin{enumerate}
	\item If the canonical solution is NIL because there is a path~$P$ in~$\P$ that is contained entirely within~$\{r_{i,1}, \ldots, r_{i,\ell-1}\}$, then since~$\ell$ is the index of the leftmost vertex from~$X'$ on~$R_i$ we have~$X' \cap \{r_{i,1}, \ldots, r_{i,\ell-1}\} = \emptyset$. Consequently, the set~$X'$ does not intersect path~$P$, contradicting the assumption that~$X'$ is a solution.
	\item Consider the case that the canonical solution is NIL because the size of~$X_{i,\ell}$ is not equal to~$b(i)$ in Step~\ref{step:toolarge}.
	\begin{enumerate}
		\item If~$|X_{i,\ell}| < b(i)$, then by the while-loop of Step~\ref{step:fillsize}, all vertices of~$\{r_{i,\ell}, \ldots, \linebreak[1] r_{i,|V(R_i)|}\}$ are in~$X_{i,\ell}$. Since~$X'_{i,\ell} = X' \cap V(R_i)$ and the leftmost vertex of~$X'_{i,\ell}$ on~$R_i$ is~$\ell$, the set~$X'_{i,\ell}$ cannot contain more vertices than~$X_{i,\ell}$. But then~$|X'_{i,\ell}| < b(i)$, showing that~$X'$ is not a solution.
		\item If~$|X_{i,\ell}| > b(i)$, Step~\ref{step:fillsize} never applied during the procedure. Consider the partition of the interval~$\{r_{i,\ell}, \ldots, r_{i,\max(X_{i,\ell})}\}$ into~$|X_{i,\ell}|$ subpaths as described in Observation~\ref{observation:partition}. Since every subpath in the partition contains exactly one vertex of~$X_{i,\ell}$, and all vertices of~$X_{i,\ell}$ are in one such subpath, it follows from~$|X_{i,\ell}| > |X'_{i,\ell}| = b(i)$ that there is such a subpath, say~$R' := \{r_{i,p}, \ldots, r_{i,q}\}$, containing no vertices of~$X'_{i,\ell}$. Since~$X'_{i,\ell}$ and~$X_{i,\ell}$ both contain~$r_{i,\ell}$, we know~$R'$ is not~$\{r_{i,\ell}\}$. Then, by Observation~\ref{observation:partition}, there is a target path in~$\P$ that is entirely contained within~$R'$. But~$X'_{i,\ell}$ contains no vertex of this path, showing that~$X'$ is not a solution.
	\end{enumerate}
\end{enumerate}
As we covered all cases that lead to the canonical solution being NIL, this concludes the proof.
\claimqed
\end{proof}

In the remainder, let~$X_{i,\ell}$ be the $\ell$-th canonical solution for~$R_i$, which is well defined by the previous claim.

\begin{claim}
$\max(X'_{i,\ell}) \leq \max(X_{i,\ell})$.
\end{claim}
\begin{proof}
Consider the process of Definition~\ref{definition:canonical}. If the loop of Step~\ref{step:fillsize} was executed at least once, then the rightmost vertex of~$R_i$ is in~$X_{i,\ell}$ and the claim is trivially true. So assume that this is not the case, and assume for a contradiction that~$\max(X'_{i,\ell}) > \max(X_{i,\ell})$. Consider the partition of the interval~$\{r_{i,\ell}, \ldots, r_{i,\max(X_{i,\ell})}\}$ as in Observation~\ref{observation:partition}. Since~$|X_{i,\ell}| = |X'_{i,\ell}| = b(i)$ and at least one vertex of~$X'_{i,\ell}$ does not lie in the interval~$\{r_{i,\ell}, \ldots, r_{i,\max(X_{i,\ell})}\}$ since~$\max(X'_{i,\ell}) > \max(X_{i,\ell})$, it follows that there is a subpath~$R'$ in the partition from which~$X'_{i,\ell}$ contains no vertices, and which therefore cannot be the subpath~$\{r_{i,\ell}\}$. As there is a path~$P \in \P$ that is entirely contained within~$R'$ by Observation~\ref{observation:partition}, the fact that~$X'_{i,\ell}$ and therefore~$X'$ contains no vertices from~$R'$ shows that~$X'$ is not a solution; contradiction.
\claimqed
\end{proof}

Using the previous claim we can finish the proof. Consider the set~$X := (X' \setminus X'_{i,\ell}) \cup X_{i,\ell}$, whose size equals that of~$X'$. We show that~$X$ is a valid solution to the instance. To see that, observe that the budget constraints are trivially satisfied since~$|X'_{i,\ell}| = |X_{i,\ell}|$. To see that all paths in~$\P$ are hit by~$X'$, consider a path~$P \in \P$. If~$P$ is hit by~$X' \setminus X'_{i,\ell}$ then it is also hit by~$X$. If~$P$ is contained entirely within~$R_i$, then by Step~\ref{step:hitpath} of Definition~\ref{definition:canonical} the path~$P$ is hit by~$X_{i,\ell}$ and thus by~$X$. If~$P$ is not contained entirely within~$R_i$ and is not hit by~$X' \setminus X'_{i,\ell}$, then it enters the petal at the leftmost or rightmost vertex of the petal and contains a prefix or suffix of the petal. (Here we use the structure of the flower graph: the interior vertices of petal~$R_i$ are not adjacent to any other vertex in the graph, only to their predecessor and successor on~$R_i$.) Since~$X'_{i,\ell}$ hits~$P$, it follows from the structure of the path in the petal that the leftmost or rightmost vertex of~$X'_{i,\ell}$ on~$R_i$ hits~$P$. But since the leftmost vertex of~$X'_{i,\ell}$ is~$r_{i,\ell}$, which is also in~$X_{i,\ell}$, and the rightmost vertex of~$X'_{i,\ell}$ does not have larger index than the rightmost vertex of~$X_{i,\ell}$ by the previous claim, it follows that~$X_{i,\ell}$ also hits~$P$. Hence all paths in~$\P$ are hit by~$X$, which is therefore a valid solution. Since the number of petals for which~$X$ does not contain a canonical solution is less than for~$X'$, by induction it follows that there is a solution for the instance whose intersection with every petal is a canonical solution.
\qed
\end{proof}

\begin{lemma} \label{lemma:createliteral}
Let~$(G,z,\P,b)$ be an instance of \HitPathsInFlower. There is a polynomial-time algorithm that, given a path~$P$ (not necessarily contained in~$\P$) which is a suffix or a prefix of a petal~$R_i$, either correctly determines that no well-defined canonical solution for~$R_i$ hits~$P$, or produces a literal of the form~$x_i \geq c$ or~$x_i \leq c$ for~$c \in \mathbb{N}_{\geq 1}$, such that the following holds.
\begin{enumerate}
	\item If~$X$ is a globally canonical solution for the instance that hits~$P$ and contains the $\ell$-th canonical solution for petal~$R_i$, then the literal is satisfied by setting~$x_i = \ell$.
	\item If~$x_i = \ell$ satisfies the literal and the $\ell$-th canonical solution~$X_{i,\ell}$ is well-defined, then~$P$ is hit by~$X_{i,\ell}$.
\end{enumerate}
\end{lemma}
\begin{proof}
The definition of the literal depends on whether~$P$ is a suffix or a prefix of a petal. First consider the case that~$P$ is a prefix of petal~$R_i$. Observe that a well-defined canonical solution~$X_{i,\ell}$ for~$R_i$ hits~$P$ if and only if~$\ell \leq \max (P)$, since~$\max(P)$ marks the index of the end of the prefix of~$R_i$ used by~$P$, and~$\ell$ is the index of the leftmost vertex of the $\ell$-th canonical solution on the petal by Observation~\ref{observation:leftmost}. Hence for this case we obtain the literal~$x_i \leq \max (P)$.

Now consider the case that~$P$ is a suffix of petal~$R_i$. The situation is similar: a well-defined canonical solution~$X_{i,\ell}$ hits the suffix~$P$ if and only if~$\max(X_{i,\ell}) \geq \min(P)$, i.e., when the rightmost vertex of the canonical solution lies right of the starting point of the suffix~$P$. Since all canonical solutions for~$R_i$ can be computed in polynomial time, we can efficiently find the indices~$\ell$, if any, for which a canonical solution is well defined satisfying~$\max(X_{i,\ell}) \geq \min(P)$. The indices for which a canonical solution is well defined form a contiguous set by Lemma~\ref{lemma:contiguous}. If~$\max(X_{i,\ell}) \geq \min(P)$ holds for some~$\ell$, then for all~$\ell' \geq \ell$ for which a canonical solution is well defined we have~$\max(X_{i,\ell'}) \geq \min(P)$ by Lemma~\ref{lemma:latersolutionsreachfurther}. Hence we can determine the smallest value~$\ell^*$ for which this holds, and find that the canonical solution on~$R_i$ hits~$P$ if and only if its index is at least~$\ell^*$. Hence we obtain the literal~$x_i \geq \ell^*$. In the case that there is no well-defined canonical solution that hits the suffix, we report this instead.

The two correctness properties follow directly from the if-and-only-if nature of our arguments above.
\qed
\end{proof}

Using the lemmata developed so far, we can present a polynomial-time algorithm for the problem in flower graphs.

\begin{theorem} \label{theorem:pathsinflower:poly}
\HitPathsInFlower can be solved in polynomial time.
\end{theorem}
\begin{proof}
We show how to reduce an instance~$(G,z,\P,b)$ with petals~$(R_1, \ldots, R_n)$ to an equivalent instance of the polynomial-time solvable \TORSTwoSat problem. The main work will be done by Lemma~\ref{lemma:createliteral} to create the literals of the formula. Let~$N := \max _{i \in [n]} |V(R_i)|$ be the maximum size of a petal. The truth value set for our multiple valued logic formula will be~$[N]$. We create a variable~$x_i$ for every petal~$i$. The clauses in the formula are produced as follows.

\begin{enumerate}
	\item For every petal index~$i \in [n]$, we compute the values of~$1 \leq \ell \leq |V(R_i)|$ for which the $\ell$-th canonical solution for petal~$R_i$ is well-defined, using the procedure of Definition~\ref{definition:canonical}. By Lemma~\ref{lemma:contiguous} these values form a contiguous interval, say~$\ell_1, \ldots, \ell_2$. We add the singleton clause~$x_i \geq \ell_1$ to the formula, as well as the singleton clause~$x_i \leq \ell_2$. If there is no well-defined canonical solution for~$R_i$ then, by Lemma~\ref{lemma:existscanonical}, the hitting set instance has no solution. In this case we simply output the answer \no.
	\item For every path~$P \in \P$ that is not contained entirely within a single petal (i.e., for every path that contains the core vertex~$z$ of the flower) we do the following. If~$P = \{z\}$ is the singleton path containing only vertex~$z$, then we output \no as a solution is not allowed to contain vertex~$z$; this path can never be hit. Otherwise, let~$P_1, P_2$ be the two connected components of~$P - \{z\}$. (In the exceptional case that~$P - \{z\}$ has only a single component because~$P$ has~$z$ as an endpoint, take~$P_1 = P_2$ to be equal to~$P - \{z\}$.) For~$k \in \{1,2\}$ let~$R_{i_k}$ be the petal containing~$P_k$ and invoke Lemma~\ref{lemma:createliteral} on~$P_k$ with~$R_{i_k}$. If the invocations for both values of~$k$ produce a literal, say~$\phi_1$ and~$\phi_2$, then add the disjunction~$\phi_1 \vee \phi_2$ as a 2-clause to the formula. If one invocation concludes that no well-defined canonical solution hits the path, but the other invocation produces a literal, then add a singleton clause with the latter literal. Finally, if neither~$P_1$ nor~$P_2$ produces a literal, then neither of the subpaths of~$P - \{z\}$ are hit by any well-defined canonical solution, and therefore the path~$P$ is not hit by any canonical solution. (Recall that solutions are forbidden to contain~$z$.)  Since, by Lemma~\ref{lemma:existscanonical}, a canonical solution exists if a solution exists at all, it follows that we can safely output \no and halt.
\end{enumerate}

The process above results in a totally ordered regular signed 2-SAT formula~$\Phi$ on~$n$ variables with~$\Oh(n + |\P|)$ clauses, which is polynomial in the size of the total input. All numbers involved are in the range~$[N]$ which is bounded by the order of the input graph~$G$. The reduction can therefore be performed in polynomial time, and produces an instance of \TORSTwoSat of polynomial size, even when encoding the numbers in unary. It remains to prove correctness of the reduction.

\begin{claim}
Formula~$\Phi$ is satisfiable if and only if~$(G,z,\P,b)$ has a solution.
\end{claim}
\begin{proof}
($\Rightarrow$) Suppose that the formula is satisfiable and consider a satisfying assignment to the variables~$x_1, \ldots, x_n$. Since the assignment satisfies the first type of clauses introduced, for every petal index~$i \in [n]$, if~$x_i = \ell$ then the $\ell$-th canonical solution for~$R_i$ is well-defined. Initialize~$X$ as an empty solution set. For each~$i \in [n]$ add the canonical solution for~$R_i$ whose index is given by~$x_i$ to the set~$X$. Since well-defined canonical solutions for~$i$ have size~$b(i)$ by Definition~\ref{definition:canonical}, this satisfies the budget constraints of the problem since the only vertices of~$R_i$ added to~$X$ are those of the canonical solution employed on that petal. As we trivially do not include~$z$ in the solution~$X$, to verify that~$X$ is a valid solution it remains to check that~$X$ intersects all paths in~$\P$. To this end, consider an arbitrary path~$P \in \P$.

\begin{enumerate}
	\item If~$P$ is contained entirely within one petal, say~$R_i$, then observe that any well-defined canonical solution for petal~$R_i$ hits~$P$ by Step~\ref{step:hitpath} of Definition~\ref{definition:canonical}. Since a canonical solution for~$R_i$ is included in~$X$, the path~$P$ is hit.
	\item If~$P$ is not contained entirely within one petal, then by the structure of flower graphs we know that~$P$ contains vertex~$z$ and was considered in the second phase of the construction. Consider the clause created on account of~$P$ during the construction above. Since the formula satisfies the clause, at least one literal is satisfied; say the literal for the subpath~$P_k$ of~$P - \{z\}$ residing in petal~$R_{i_k}$. Then Lemma~\ref{lemma:createliteral} guarantees that the canonical solution employed on~$R_{i_k}$ hits~$P_k$, and therefore hits the larger path~$P$ as well.
\end{enumerate}

\noindent As~$X$ hits all paths in~$\P$, this proves the forward direction.

($\Leftarrow$) For the reverse direction, suppose that~$(G,z,\P,b)$ has a solution. By Lemma~\ref{lemma:existscanonical} there is a globally canonical solution~$X$. For every petal index~$i$ let~$\ell_i$ be such that~$X$ includes the $\ell_i$-th canonical solution on~$R_i$, and assign variable~$x_i$ the value~$\ell_i$. Let us check that this assignment satisfies the formula. Every clause of the first type is satisfied by any setting corresponding to the index of a canonical solution, which is clearly the case. For the clauses of the second type that are produced on account of paths~$P \in \P$, observe that~$X$ intersects such a path in a connected component of~$P - \{z\}$, since~$z \not \in X$. By Lemma~\ref{lemma:createliteral} this implies that the corresponding literal of the clause is satisfied, implying that the entire clause is satisfied. Hence all types of clauses are satisfied, showing the formula to be satisfiable.
\claimqed
\end{proof}

The claim shows that to solve the hitting set problem, it suffices to check the satisfiability of the polynomial-sized \TORSTwoSat instance. As the latter can be done in polynomial time, this proves Theorem~\ref{theorem:pathsinflower:poly}.
\qed
\end{proof}

\subsection{Hitting Paths in Graphs} \label{subsection:hit:paths}

In this section we will show that an instance of \HitPathsInGraph can be reduced to~$2^{5k}$ instances of \HitPathsInFlower. By the results of the previous section, this leads to an FPT algorithm. 

We will frequently use the following observation. It formalizes that if~$v$ is a degree-one vertex in~$G$ and we are looking for a set that hits all paths in~$\P$, then either there is a single-vertex path~$P = \{v\} \in \P$, forcing~$v$ to be in any solution, or there is an optimal solution that does not contain~$v$.

\begin{observation} \label{observation:removeleaf}
Let~$(G,k,t,\P)$ be an instance of \HitPathsInGraph and let~$v \in V(G)$ have degree at most one. 
\begin{enumerate}
	\item If the singleton path~$P = \{v\}$ is contained in~$\P$, then~$(G,k,t,\P)$ is equivalent to the instance obtained by decreasing~$t$ by one, removing~$v$ from the graph, and removing all paths containing~$v$ from~$\P$.
	\item Otherwise,~$(G,k,t,\P)$ is equivalent to the instance obtained by removing~$v$ from the graph and replacing every path~$P \in \P$ by~$P \setminus \{v\}$.
\end{enumerate}
The cyclomatic number is not affected by these operations.
\end{observation}

For an instance~$(G,\P,k,t)$ of \HitPathsInGraph and a vertex subset~$S \subseteq V(G)$, the \emph{cost of the subgraph induced by~$S$}, denoted~$\opt(S)$, is defined as the minimum cardinality of a set that hits all paths~$P \in \P$ for which~$V(P) \subseteq S$. Equivalently,~$\opt(S)$ is the minimum cardinality of a set that hits all paths~$\{ P \in \P \mid V(P) \subseteq S\}$ in the graph~$G[S]$. Observe that if~$S$ induces an acyclic subgraph of~$G$, then this value is computable in polynomial time as discussed in Section~\ref{section:prelims}. To reduce the general \HitPathsInGraph problem to the version with budget constraints discussed in the previous section, the following lemma is useful for determining relevant values for the budgets.

\begin{lemma} \label{lemma:budget:on:paths}
Let~$(G,\P,k,t)$ be an instance of \HitPathsInGraph. Let~$S$ be the vertices of degree unequal to two in~$G$. There is a minimum-size hitting set~$X$ for~$\P$ such that, for every connected component~$C$ of~$G - S$, we have~$\opt(C) \leq |X \cap V(C)| \leq \opt(C)+1$.
\end{lemma}
\begin{proof}
The fact that~$X \cap V(C) \geq \opt(C)$ for all hitting sets of~$\P$ follows trivially since~$X \cap V(C)$ is a hitting set for the induced subinstance. For the other inequality we exploit the structure of the graph.

Let~$X$ be a minimum-size hitting set for~$\P$, whose size may be less than~$t$. We give a proof by induction on the number of components for which~$X \cap V(C)$ exceeds~$\opt(C) + 1$. The statement is trivially true if this number is zero. Otherwise, fix a component~$C$ for which~$X \cap V(C) > \opt(C) + 1$. As~$C$ is a connected subgraph containing only vertices of degree two, the neighborhood~$N_G(C)$ has size at most two, and is contained within~$S$. Let~$X_C$ be a minimum-cardinality hitting set for the instance induced by~$C$, of size~$\opt(C)$. Consider the set~$X' := (X \setminus V(C)) \cup N_G(C) \cup X_C$, whose intersection with~$C$ has size~$\opt(C)$. Since~$|N_G(C) \cup X_C| \leq \opt(C) + 2$ and~$X \cap V(C) \geq \opt(C) + 2$, the set~$X'$ is not bigger than~$X$. We show it to be a hitting set as well. To see that, observe that all paths contained in~$G[C]$ are hit by~$X_C$, as it is a solution to the subproblem induced by~$C$. Any path intersecting~$C$ that is not hit by~$X_C$ was not included in the subinstance induced by~$C$, and hence contains a vertex of~$N_G(C)$. Such paths are therefore hit by~$X'$. All paths that do not intersect~$C$ are hit by~$X \setminus V(C)$, and are therefore hit by its superset~$X'$ as well. It follows that~$X'$ is a hitting set of minimum cardinality. As the number of components~$C$ from which it uses at least~$\opt(C) + 2$ vertices is strictly smaller than for~$X$ (the vertices of~$N_G(C)$ belong to~$S$ and therefore do not increase this number, as components are taken of~$G-S$), the proof now follows by induction.
\qed
\end{proof}

Using these ingredients we give an algorithm for \HitPathsInGraph.

\begin{theorem} \label{theorem:pathsingraph:fpt}
\HitPathsInGraph parameterized by cyclomatic number can be solved in time~$2^{5k} \cdot (|G| + |\P|)^{\Oh(1)}$.
\end{theorem}
\begin{proof}
When presented with an input~$(G,\P,k,t)$, the algorithm proceeds as follows. First, as a preprocessing step, the algorithm repeatedly removes vertices of degree at most one from the graph using Observation~\ref{observation:removeleaf}. If the resulting graph is empty, then we can simply decide the problem: the answer is \yes if and only if the value of~$t$ was not decreased below zero by these operations. Otherwise we obtain a graph with minimum degree at least two. While this graph is disconnected, add an arbitrary edge between two distinct connected components. This does not change the answer to the instance (the paths~$\P$ to be hit are unchanged) and leaves the cyclomatic number unchanged. From now on we therefore assume that the instance we work with has minimum degree at least two and consists of a connected graph. For ease of notation, we refer to instance resulting from these steps simply as~$(G,\P,k,t)$. If~$G$ consists of just a simple cycle (i.e.,~$G$ is 2-regular) then we can decide the problem in polynomial time as discussed in Section~\ref{section:prelims}, so we focus on the case that the set~$S$ of vertices of degree at least three is nonempty. By Proposition~\ref{proposition:degtwo}, the size of~$S$ is bounded by~$2k$. The main idea of the algorithm is to use branching make two successive guesses. 
\begin{itemize}
	\item First, we guess which vertices from~$S$ are used in a solution. Concretely, we try all subsets~$S' \subseteq S$ and test whether there is a solution~$X$ for which~$X \cap S = S'$. 
	\item For every such set~$S'$, we do the following. By Lemma~\ref{lemma:budget:on:paths}, there is a minimum-size hitting set that intersects every component~$C$ of~$G - S$ (which is a path) in either~$\opt(C)$ or~$\opt(C) + 1$ vertices. Let~$\C$ denote the set of these components. By Proposition~\ref{proposition:numcomponents}, we have~$|\C| < k + |S| \leq 3k$. We now guess the collection~$\C' \subseteq \C$ of components~$C$ for which the solution uses~$\opt(C)$ vertices. Every guess~$\C'$ defines a budget~$b(C)$ for each component~$C \in \C$ as follows: $b(C) = \opt(C)$ if~$C \in \C'$, and~$b(C) = \opt(C) + 1$ otherwise.
\end{itemize}

Having guessed both~$S'$ and~$\C'$, we create an instance of \HitPathsInFlower to verify whether there is a hitting set~$X$ for the paths~$\P$ such that~$X \cap S = S'$ and for all components~$C$ of~$G-S$ we have~$|X \cap C| = b(C)$. Observe that these constraints on~$X$ completely determine its size, which must be~$|S'| + \sum _{C \in \C} b(C)$. Hence if the size exceeds~$t$, then these guessed sets will not lead to a hitting set of the desired size, and can therefore be skipped. When we have a guess that leads to a hitting set size of at most~$t$, we aim to produce an instance of \HitPathsInFlower to check whether there is a solution consistent with the guesses. To this end, initialize~$G'$ as a copy of~$G$, and~$\P'$ as a copy of~$\P$. We modify these structures to create an input on a flower graph. Throughout these modifications there will be a clear correspondence between components of~$G' - S$ and those of~$\C$, so that we may refer to the budgets of components~$C$ of~$G' - S$. For each guess~$S'$ and~$\C'$, we proceed as follows.

\begin{enumerate}
	\item Remove all the vertices of~$S'$ from the graph~$G'$ and remove all paths hit by~$S'$ from~$\P'$.\label{flower:step:removeshitpaths}
	\item For all paths~$P \in \P'$ for which there is a component~$C$ of~$G' - S$ such that all vertices of~$C$ belong to~$P$ and~$b(C) > 0$, remove~$P$ from the set~$\P'$. All hitting sets that contain~$b(C)$ vertices from~$C$ must hit~$P$, so we can drop the constraint~$P$ because we will introduce a budget constraint on~$C$.\label{flower:step:removebudgethitpaths}
	\item For all components~$C$ of~$G' - S$ such that~$b(C) = 0$, do the following. Remove the vertices of~$C$ from the graph~$G'$. For every~$P \in \P'$, replace~$P$ by the subgraph~$P - V(C)$. This may cause the elements of~$\P'$ to become disconnected subgraphs, rather than paths, but this will be resolved in the next step.\label{flower:step:splitpaths}
	\item The final step identifies several vertices in the graph into a single core vertex, to obtain a flower structure. Concretely, update the graph~$G'$ by identifying all vertices of~$S \setminus S'$ into a single vertex~$z$. Similarly, update every subgraph~$P \in \P'$ by identifying all vertices of~$(S \setminus S') \cap V(P)$ into a single vertex~$z$.\label{flower:step:merge}
\end{enumerate}

\begin{figure}[t]
\begin{center}
\subfigure[]{\label{fig:toflower1}
\includegraphics{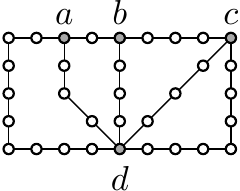}
}
\subfigure[]{\label{fig:toflower4}
\includegraphics{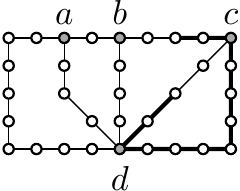}
}
\subfigure[]{\label{fig:toflower3}
\includegraphics{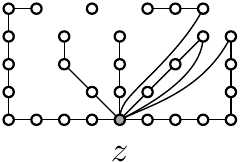}
}
\subfigure[]{\label{fig:toflower5}
\includegraphics{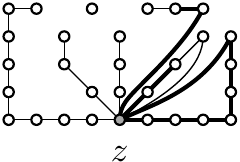}
}
\caption{(\ref{fig:toflower1}) A graph with cyclomatic number~$5$, whose vertices of degree~$\geq 3$ are~$S := \{a,b,c,d\}$. (\ref{fig:toflower4}) A simple path~$P$ in the graph. (\ref{fig:toflower3}) Illustration of reduction Steps~\ref{flower:step:removeshitpaths} and~\ref{flower:step:merge} in the algorithm for the guess~$S' = \{a,b\}$. Vertices~$a$ and~$b$ are deleted, while~$c$ and~$d$ are identified into a single vertex~$z$ to obtain a flower structure. (\ref{fig:toflower5}) Merging~$c$ and~$d$ into~$z$ turns~$P$ into a cyclic subgraph~$P'$. The bottom right petal is contained entirely within~$P'$. If its budget is positive, any solution hits~$P'$ in that petal, causing~$P$ to be removed in Step~\ref{flower:step:removebudgethitpaths}. If its budget is zero, the vertices of the petal are removed from~$P'$ instead in Step~\ref{flower:step:splitpaths}, to eliminate the cycle.}\label{fig:mergeintoflower}
\end{center}
\end{figure}

Let~$G^*, \P^*$ denote the resulting graph and system of subgraphs. Refer to Figure~\ref{fig:mergeintoflower} for an illustration of these steps.

\begin{numberedclaim}
$G^*$ is a flower with core~$z$ and all subgraphs in~$\P^*$ are simple paths.
\end{numberedclaim}
\begin{proof}
Let us first verify that~$G^*$ is indeed a flower. As every vertex of~$S$ was either removed or merged into~$z$, we find that~$G^* - \{z\}$ is a subgraph of~$G - S$. Since~$S$ is the set of vertices of degree unequal to two, and~$G$ had no vertices of degree at most one after preprocessing, every connected component of~$G - S$ consists of vertices that have degree two in~$G$ and therefore such components form paths. As~$S$ is not empty, every such component has exactly two neighbors in~$S$, which are adjacent to the first and last vertex of the path. Hence no interior vertex of such paths is adjacent to~$S$. Since~$G^* - \{z\}$ is a subgraph of~$G - S$ it follows that all connected components of~$G^* - \{z\}$ are paths and no interior vertex of such a path is adjacent to~$z$. Hence~$G^*$ is a flower with core~$z$.

We continue by proving the second part of the claim. Consider a subgraph~$P^*$ in the final set~$\P^*$, and let~$P \in \P$ be the simple path in~$G$ from which it originated. As~$P^*$ is present after the last step, it follows that~$P$ did not meet the precondition for removal in Step~\ref{flower:step:removeshitpaths} and therefore~$P \cap S' = \emptyset$, showing that~$P$ is a simple path in~$G - S'$. Similarly, as~$P$ was not removed by Step~\ref{flower:step:removebudgethitpaths} we know that~$P$ does not fully contain any component~$C$ of~$G - S$ with positive budget. Consider what happens when deleting components~$C$ with budget zero in Step~\ref{flower:step:splitpaths}. Since all vertices of~$G - S$ have degree two in~$G$, there are at most two components~$C$ of~$G - S$ from which~$P$ contains at least one, but not all vertices: these are the components containing the endpoints of~$P$. Hence if~$P$ is transformed into a disconnected subgraph~$P'$ by Step~\ref{flower:step:splitpaths}, then (1)~$P'$ contains at least one vertex of~$S \setminus S'$, and (2) there are at most two connected components in~$P' - S$, and both these components have (in subgraph~$P'$) a neighbor in~$S \setminus S'$. This shows that when all vertices of~$S \setminus S'$ are merged into a single vertex~$z$ by Step~\ref{flower:step:merge}, then the disconnected subgraph~$P'$ is transformed into a path~$P^*$ containing the core vertex~$z$ in its interior. Hence all subgraphs~$P^*$ in~$\P^*$ are simple paths.
\claimqed
\end{proof}

The claim shows that we can use the structures~$G^*$ and~$\P^*$ resulting from the process above to formulate an instance of \HitPathsInFlower. To that end, we use~$G^*$ as the flower graph,~$z$ as the core, and~$\P^*$ as the set of paths to be hit. We number the connected components of~$G^* - \{z\}$, which are the petals of the flower, as~$R_1, \ldots, R_n$. Each such petal corresponds to a connected component of~$G - S$ for which we assigned a budget when guessing~$\C' \subseteq \C$; we define the budget function~$b^*$ for the instance by letting~$b^*(i)$ be~$b(C_i)$ where~$C_i$ is the component of~$G - S$ corresponding to~$R_i$. This results in a valid instance~$(G^*,z,\P^*,b^*)$ of \HitPathsInFlower. For the correctness of the algorithm, the following claim is crucial.

\begin{numberedclaim}
For every guess of~$S' \subseteq S$ and~$\C' \subseteq \C$, the following are equivalent. 
\begin{enumerate}
	\item There is a hitting set~$X$ for the paths~$\P$ in graph~$G$ such that~$X \cap S = S'$ and all components~$C$ of~$G-S$ satisfy~$|X \cap C| = b(C)$.\label{eqv:graphhitset}
	\item The produced instance~$(G^*,z,\P^*,b^*)$ of \HitPathsInFlower has a solution.\label{eqv:flowerhitset}
\end{enumerate}
\end{numberedclaim}
\begin{proof}
(\ref{eqv:graphhitset}$\Rightarrow$\ref{eqv:flowerhitset}) Suppose there is a hitting set~$X$ satisfying the stated conditions. We claim that~$X^* := X \setminus S$ is a solution for~$(G^*,z,\P^*,b^*)$. By the preconditions, set~$X^*$ satisfies the budget constraints for the petals. Consider a path~$P^* \in \P^*$ derived from a path~$P \in \P$. Set~$X$ does not hit~$P$ in a component of~$G - S$ with~$b(C) = 0$, as~$X$ contains no vertices of such components. Set~$X$ does not intersect~$P$ in a vertex of~$S'$, as such paths are not present in~$\P^*$ due to Step~\ref{flower:step:removeshitpaths}. It follows that~$X$ intersects~$P$ in a vertex~$v$ of~$V(G) \setminus S$ that lies in a connected component of~$G - S$ with positive budget. Hence the component forms a petal~$R_i$ in~$G^*$, and the intersection of~$P$ with the petal is contained in~$P^*$. Hence~$X^*$ hits~$P^*$ at~$v$. Since all paths in~$\P^*$ are hit, the budget constraints are met, and~$z \not \in X^*$, it follows that~$X^*$ is a solution to~$(G^*,z,\P^*,b^*)$.

(\ref{eqv:flowerhitset}$\Rightarrow$\ref{eqv:graphhitset}) For the reverse direction, suppose the flower problem has a solution~$X^*$. We claim that~$X := X^* \cup S'$ is a hitting set for the paths~$\P$ in~$G$ with~$|X \cap C| = b(C)$ for all components~$C$ of~$G - S$. The latter condition is easily verified: components~$C$ with~$b(C) = 0$ were discarded in Step~\ref{flower:step:splitpaths}, do not occur in~$G^*$, and therefore~$X^*$ contains no vertices of such components. For components with positive budget, which also exist in~$G^*$, the budget constraints in the definition of \HitPathsInFlower ensure that~$|X \cap C| = |X^* \cap C| = b(C)$. Let us verify that~$X$ indeed hits all paths~$\P$ in~$G$ by considering an arbitrary~$P \in \P$. If~$P$ contains a vertex of~$S'$ then~$X$ trivially hits~$P$. Similarly, if~$P$ fully contains a component~$C$ of~$G - S$ with~$b(C) > 0$, then~$X^*$ contains at least one vertex of~$C$ and therefore of~$P$; hence~$X \supseteq X^*$ also hits~$P$. In the remaining case, the construction of~$G^*,\P^*$ ensures that~$\P^*$ contains a path~$P^*$ such that~$P^* - \{z\}$ is a subgraph of~$P - S$. Since~$X^*$ hits~$P^*$ in a vertex other than~$z$, this vertex is included in~$X$ and therefore~$X$ hits~$P$. It follows that~$X$ is indeed a hitting set for all paths in~$\P$, which concludes the proof of the equivalence.
\claimqed
\end{proof}

Using the claim, the final part of the algorithm becomes clear. For every guess~$S' \subseteq S$ and~$\C' \subseteq \C$ that leads to a solution of size at most~$t$, we construct the corresponding instance of \HitPathsInFlower and solve it using Theorem~\ref{theorem:pathsinflower:poly}. Since the flower instances are not larger than the input instance, this can be done in time~$(|G| + |\P|)^{\Oh(1)}$ for every guess. As there are~$2^{|S|} \cdot 2^{|\C|} \leq 2^{2k} \cdot 2^{3k}$ options for~$S'$ and~$\C'$ to check, the total running time is bounded by~$2^{5k} (|G| + |\P|)^{\Oh(1)}$. If one of the \HitPathsInFlower instances has answer \yes, then we output \yes; otherwise we output \no. In one direction, the correctness of this approach follows from the previous claim together with the facts that flower instances are only produced when the size of the resulting hitting sets is at most~$t$. For the other direction, if~$(G,\P,t,k)$ has a hitting set of size at most~$t$, then by Lemma~\ref{lemma:budget:on:paths} there is a minimum-cardinality hitting set~$X$ (whose size is at most~$t$) whose intersection with every component~$C$ of~$G - S$ is either~$\opt(C)$ or~$\opt(C) + 1$. In the branch where~$S' = X \cap S$ and~$\C'$ consists of the components where we use~$\opt(C)$ vertices, this leads to a \yes-instance of \HitPathsInFlower. This concludes the proof of Theorem~\ref{theorem:pathsingraph:fpt}.
\qed
\end{proof}

We remark that, while the previous theorem shows that \HitPathsInGraph is fixed-parameter tractable parameterized by the cyclomatic number, this problem is unlikely to admit a polynomial kernel. The general \HittingSet problem parameterized by the number of universe elements~$n$ can be reduced to an instance of \HitPathsInGraph with cyclomatic number~$\Oh(n^2)$: if we let~$G$ be a complete~$n$-vertex graph, which has cyclomatic number~$\Oh(n^2)$, then we can model any subset of the universe as a simple path in~$G$. Hence there is a polynomial-parameter transformation from \HittingSet parameterized by the universe size to \HitPathsInGraph parameterized by cyclomatic number. Since \HittingSet parameterized by universe size has no polynomial kernel unless \containment~\cite[Theorem 5.3]{DomLS14}, the same holds for \HitPathsInGraph parameterized by cyclomatic number.

\section{Hardness proofs} \label{section:hardness}

In this section we develop several hardness proofs. It turns out to be convenient to first prove the W[1]-hardness of 3-SAT in multiple valued logic. A similar result concerning the W[1]-hardness of \emph{not-all-equal} 3-SAT was obtained independently by Bringmann et al.~\cite{BringmannHML15}, who studied the problem under the name \textsc{NAE-Integer-$3$-SAT}.

\begin{theorem} \label{thm:threesat:whard}
The problem \SignedThreeSatn is W[1]-hard.
\end{theorem}
\begin{proof}
To establish the theorem we give an FPT-reduction from the W[1]-complete \kClique problem~\cite[Chapter 21]{DowneyF13}. Let~$(G,k)$ be an instance of \kClique, asking whether the graph~$G$ has a clique of size~$k$. We use an edge representation strategy to encode this problem into an instance of \nTORSThreeSat whose parameter, the number of variables~$n$ in the formula, is~$\Oh(k^2)$. We may assume that~$|E(G)| \geq \binom{k}{2} \geq k \geq 2$, as the instance is trivial otherwise. We may also assume that~$G$ has no isolated vertices. The formula is constructed as follows.

There are variables~$x_1, \ldots x_k$ corresponding to a choice of~$k$ vertices in the clique. In addition, there are~$\binom{k}{2}$ variables~$x_{i,j}$ for~$1 \leq i < j \leq k$ that correspond to the edges between these vertices. The truth value set for the formula is the range of integers from~$1$ to~$|E(G)|$, so~$N := |E(G)|$. Number the vertices of~$G$ as~$v_1, \ldots, v_{|V(G)|}$, and the edges from~$1$ to~$N$, arbitrarily. For every~$1 \leq i < j \leq k$ and for every possible edge index~$\ell \in [N]$, we add four clauses to the formula. Let~$\{v_p,v_q\}$ be the endpoints of the $\ell$-th edge such that~$p < q$. We add the following clauses:
	$$(x_{i,j} \leq \ell-1 \vee x_{i,j} \geq \ell+1 \vee x_i \leq p) \quad (x_{i,j} \leq \ell-1 \vee x_{i,j} \geq \ell+1 \vee x_i \geq p),$$
	$$(x_{i,j} \leq \ell-1 \vee x_{i,j} \geq \ell+1 \vee x_j \leq q) \quad (x_{i,j} \leq \ell-1 \vee x_{i,j} \geq \ell+1 \vee x_j \geq q).$$
	To obtain a valid formula, we omit the literal~$x_{i,j} \leq \ell - 1$ when~$\ell=1$, as do we omit the literal~$x_{i,j} \geq \ell + 1$ when~$\ell = N$. These clauses are automatically satisfied if~$x_{i,j} \neq \ell$, i.e., if~$x_{i,j}$ does \emph{not} select the~$\ell$-th edge. If~$x_{i,j} = \ell$, however, then the clauses force~$x_i$ to have the value~$p$ and~$x_j$ to have value~$q$.
	
The conjunction of the produced clauses for all valid values of~$i,j$, and~$\ell$ forms the output formula. The construction can be performed in polynomial time and produces an instance of \nTORSThreeSat whose parameter~$n$ is~$\binom{k}{2} + k \in \Oh(k^2)$, which is suitably bounded. To complete the proof it suffices to show that~$G$ has a $k$-clique if and only if the formula is satisfiable.

\begin{claim}
If~$G$ has a $k$-clique, then the formula is satisfiable.
\end{claim}
\begin{proof}
Consider a $k$-clique in~$G$ and let the indices of its vertices be~$u_1, u_2, \ldots, u_k$ in order of increasing value. For~$i \in [k]$ assign variable~$x_i$ value~$u_i$, and for~$1 \leq i < j \leq k$ assign variable~$x_{i,j}$ the value of the index of the edge between~$v_{u_i}$ and~$v_{u_j}$. As observed above, the clauses that are created for values~$i,j,\ell$ such that~$x_{i,j}$ does not select the~$\ell$-th edge, are satisfied. It is easy to verify that when~$x_{i,j} = \ell$, the third literal of the created clauses is satisfied. Hence all clauses are satisfied and the formula is satisfiable.
\claimqed
\end{proof}

\begin{claim}
If the formula is satisfiable, then~$G$ has a $k$-clique.
\end{claim}
\begin{proof}
Consider an assignment to the variables that satisfies all clauses. Consider the values taken by the variables~$x_1, \ldots, x_k$. Suppose that some variable~$x_i$ with~$i<k$ has a value exceeding~$|V(G)|$. Then consider the value~$\ell$ of variable~$x_{i,i+1}$, and the clauses produced for this combination. Since~$x_{i,i+1} \leq \ell - 1$ is false, as is~$x_{i,i+1} \geq \ell+1$, we must have~$x_i \leq p$ where~$p$ is the lowest-indexed endpoint of the $\ell$-th edge; but this contradicts the assumption that~$x_i > |V(G)|$. A similar contradiction is reached when~$x_i = k$ by considering variable~$x_{i-1, i}$ instead. Hence the variables~$x_1, \ldots, x_k$ represent indices of vertices in~$G$. 

Next, assume for a contradiction that there are indices~$1 \leq i < j \leq k$ such that~$x_i = x_j$, and let~$\ell$ be the value of variable~$x_{i,j}$. As observed above, the clauses added for the combination~$i,j,\ell$ are only satisfied if~$x_i$ and~$x_j$ represent the indices of the endpoints of the~$\ell$-th edge. But since~$G$ is a simple graph without self-loops, these indices are distinct and therefore these clauses cannot all be satisfied if~$x_i$ and~$x_j$ coincide. Hence the variables~$x_1, \ldots, x_k$ take~$k$ distinct values in the range of~$1$ to~$|V(G)|$.

We claim that the~$k$ vertices in~$G$ whose indices correspond to the values of~$x_1, \ldots, x_k$ form a clique. To see that all pairs of these vertices are adjacent in~$G$, consider a pair~$1 \leq i < j \leq k$ and the value~$\ell$ taken by variable~$x_{i,j}$. The clauses produced for~$i,j,\ell$ are only satisfied if~$x_i$ is the lower-indexed endpoint of the $\ell$-th edge and~$x_j$ is the higher-index endpoint of that edge. Given the values of~$x_i$ and~$x_j$, the clauses can therefore only be satisfied if~$\ell$ is the index of the edge between vertices with indices~$x_i$ and~$x_j$. Hence the edge connecting this pair must be present in~$G$. As~$i$ and~$j$ were arbitrary, this shows that all vertex pairs are adjacent. Hence the set of vertices with indices~$x_1, \ldots, x_k$ is a $k$-clique in~$G$.
\claimqed
\end{proof}

This concludes the proof of Theorem~\ref{thm:threesat:whard}.
\qed
\end{proof}


Theorem~\ref{thm:threesat:whard} is used as the starting point for the next hardness proof.

\begin{theorem} \label{theorem:hitclaws:whard}
It is W[1]-hard to determine, given a graph~$G$ with cyclomatic number~$k$, a set~$\S$ of subgraphs of~$G$, each isomorphic to a tree with at most three leaves, and an integer~$t$, whether there is a set of~$t$ vertices in~$G$ that intersects all subgraphs in~$\S$.
\end{theorem}
\begin{proof}
We give an FPT-reduction from \nTORSThreeSat. Consider an instance of that problem, consisting of a signed 3-CNF formula over variables~$x_1, \ldots, x_n$ whose truth value set is~$[N]$. We assume that there are no clauses that are trivially satisfied (that contain literals~$x_i \leq c_1$ and~$x_i \geq c_2$ for~$c_2 \leq c_1 + 1$), as they can be efficiently recognized and removed without changing the answer.

We construct a hitting set problem on a flower graph~$G$ that has a core~$z$ and~$n$ petals~$R_1, \ldots, R_n$. Each petal is a path on~$N$ vertices whose endpoints are adjacent to~$z$. It is easy to see that this gives a cyclomatic number of at most~$k = n$ for the graph~$G$, as removing the~$n$ edges from~$z$ to the last vertex of each petal gives an acyclic graph. We seek a hitting set of size at most~$t := n$.

Signed literals of the formula have the form~$x_i \leq c$ or~$x_i \geq c$ for~$c \in [N]$. We associate every literal to a prefix or suffix of a petal: a literal~$x_i \leq c$ corresponds to the prefix~$\{r_{i,1}, \ldots, r_{i,c}\}$ of petal~$R_i$, while a literal~$x_i \geq c$ corresponds to the suffix~$\{r_{i,c}, \ldots, r_{i,N}\}$. For every clause~$C$ of the formula, we consider the pre/suffixes associated to its literals. We add the subgraph~$S_C$ that is induced by their vertices, together with~$z$, to the set~$\S$ of subgraphs to be hit. Observe that, since there are no clauses that are trivially satisfied, each such subgraph~$S_C$ induces a tree in~$G$ with at most three leaves. In addition, for every petal~$R_i$ we add the path~$R_i$ as a subgraph to~$\S$. This concludes the description of the hitting set instance.

\begin{numberedclaim} \label{claim:hitset:iff:sat}
There is a hitting set of size at most~$t$ if and only if the formula is satisfiable.
\end{numberedclaim} 
\begin{proof}
($\Rightarrow$) Suppose there is a hitting set~$X$ of size~$t = n$. Since every petal~$R_i$ is present as a subgraph in~$\S$ that must be hit, and the petals are pairwise disjoint, it follows that~$X$ contains exactly one vertex of each petal. In particular, the core~$z$ is not in~$X$. Consider the assignment that sets the value of variable~$x_i$ to the index of the vertex in~$X \cap V(R_i)$, which is a number in the range~$[N]$. To see that an arbitrary clause~$C$ is satisfied, consider the subgraph~$S_C$ created on account of the clause, which consists of~$z$ together with at most three pre/suffixes of petals, one for each literal of~$C$. As the pre/suffix that is hit by~$X$ corresponds to a literal that is satisfied by the assignment, clause~$C$ is satisfied. As~$C$ was arbitrary, the formula is satisfiable.

($\Leftarrow$) Suppose that the formula is satisfied by a particular assignment to~$x_1, \ldots, \linebreak[1] x_n$. Let~$X$ contain vertex~$r_{i, x_i}$ for all~$i \in [n]$. Then all petals are hit by~$X$, and all subgraphs~$S_C$ added on account of a clause~$C$ are hit at a pre/suffix corresponding to a literal in the clause that is satisfied.
\claimqed
\end{proof}

The claim shows the correctness of the reduction. It is a valid FPT-reduction since it can be executed in polynomial time and the new parameter~$k$ equals the old parameter~$n$. Since \SignedThreeSatn is W[1]-hard by Theorem~\ref{thm:threesat:whard}, this concludes the proof.
\qed
\end{proof}

By slightly modifying the construction, we can also obtain the following result which shows that hitting paths in graphs is para-NP-complete~\cite{FlumG06} parameterized by the feedback vertex number of the graph.

\begin{theorem} \label{theorem:hitpaths:fvs:npc}
It is NP-complete to determine, given a graph~$G$ with a feedback vertex set of size two, a set~$\P$ of simple paths in~$G$, and an integer~$t$, whether there is a set of~$t$ vertices in~$G$ that intersects all paths in~$\P$.
\end{theorem}
\begin{proof}
The proof is similar to that of Theorem~\ref{theorem:hitclaws:whard}, so we only mention the key points. An instance of \SignedThreeSatn on variables~$x_1, \ldots, x_n$ with truth value set~$[N]$ is reduced to an instance of the hitting set problem as follows. For every variable~$x_i$ we create a new path~$R_i$ on~$N$ vertices in the graph. Finally we add two universal vertices~$z, z'$ to the graph, adjacent to all vertices on all created paths. The resulting graph~$G$ has a feedback vertex set of size two, being~$\{z,z'\}$. For every~$i \in [n]$ we add~$R_i$ to~$\S$ to ensure that a vertex of~$R_i$ is selected in every hitting set. For every clause of the formula, we consider the (at most) three pre/suffixes of the petals~$R_i$ corresponding to its literals, as in Theorem~\ref{theorem:hitclaws:whard}. Since both~$z$ and~$z'$ are universal vertices, there is a simple path in~$G$ consisting of the first pre/suffix, vertex~$z$, the second pre/suffix, the vertex~$z'$, and ending with the last pre/suffix. For every clause we add such a path to~$\P$, which ensures that the clause must be satisfied when all paths are hit. Finally, we set the budget to~$t := n$ to ensure that valid solutions select one value for each variable. Following the argumentation of Theorem~\ref{theorem:hitclaws:whard} it is easy to see that the reduction is correct. Since the \SignedThreeSatn problem is NP-complete, the theorem follows.
\qed
\end{proof}

We close this section on hardness by a discussion of subexponential-time algorithms. The construction in Theorem~\ref{theorem:hitpaths:fvs:npc} can be used to reduce an $n$-variable instance of the classical 3-SAT problem (with binary variables) to the problem of hitting simple paths in a graph of cyclomatic number~$\Oh(n)$. This implies that, assuming the exponential-time hypothesis~\cite{ImpagliazzoPZ01}, the dependence on~$k$ in Theorem~\ref{theorem:pathsingraph:fpt} cannot be improved to~$2^{o(k)}$.

\section{Conclusion}

We have analyzed the problem of hitting subgraphs of a restricted form within a larger host graph, parameterized by structural measures of the host graph. There are several research directions related to this work that remain unexplored. For example, we have not touched upon the issue of computing, given a generic hitting set instance consisting of a set system~$\F$ over a universe~$U$, how complex graphs on vertex set~$U$ must be in which every set in~$\F$ induces a connected subgraph. What is the complexity of finding, given~$\F$ and~$U$, a graph of minimum cyclomatic number that embeds~$\F$ in this way? Alternatively, what is the complexity of finding the minimum cyclomatic number of a graph~$G$ such that for every set~$S \in \F$, there is a simple path in~$G$ on vertex set~$S$? Efficient algorithms for this task 
could be used to transform generic hitting set instances into inputs of \HitPathsInGraph, on which Theorem~\ref{theorem:pathsingraph:fpt} can be applied.

One can also consider aggregate parameterizations of the hitting set problem using the measure of structure introduced here. We have shown that \HitPathsInGraph is FPT parameterized by the cyclomatic number. It is well known that the general \HittingSet problem is FPT parameterized by the number of sets, as it can be solved by dynamic programming. Suppose we have a \HittingSet instance where there are~$k_1$ arbitrary sets, and there is a graph~$G$ of cyclomatic number~$k_2$ such that the remaining sets correspond to paths in~$G$. Is \HittingSet parameterized by~$k_1+k_2$ FPT, when this structure is given?

The complexity of the problem changes significantly when weights are introduced for the elements in the universe and the task is to find a minimum-weight hitting set. A simple reduction from \VertexCover shows that finding a minimum-weight set that hits a prescribed set of three-vertex paths in a star graph is already NP-complete. This suggests some topics for further investigation; we list some examples.

\begin{enumerate}
	\item Is the problem of finding a minimum-weight vertex set that hits a prescribed set of \emph{directed paths} in a \emph{directed tree} polynomial-time solvable?
	\item What is the parameterized complexity of the problem of hitting weighted paths in a tree plus~$k$ edges, when the largest weight value is bounded by a constant?
\end{enumerate}

\textbf{Acknowledgments}. We are grateful to Mark de Berg and Kevin Buchin for interesting discussions that triggered this research.

\bibliographystyle{abbrvurl}
\bibliography{../Paper}

\begin{thebibliography}{10}

\bibitem{Abu-Khzam10}
F.~N. Abu-Khzam.
\newblock A kernelization algorithm for d-{H}itting set.
\newblock {\em J. Comput. Syst. Sci.}, 76(7):524--531, 2010.
\newblock \href {http://dx.doi.org/10.1016/j.jcss.2009.09.002}
  {\path{doi:10.1016/j.jcss.2009.09.002}}.

\bibitem{AspvallPT79}
B.~Aspvall, M.~F. Plass, and R.~E. Tarjan.
\newblock A linear-time algorithm for testing the truth of certain quantified
  boolean formulas.
\newblock {\em Information Processing Letters}, 8(3):121--123, 1979.
\newblock \href {http://dx.doi.org/10.1016/0020-0190(79)90002-4}
  {\path{doi:10.1016/0020-0190(79)90002-4}}.

\bibitem{BejarHM01}
R.~B{\'{e}}jar, R.~H{\"{a}}hnle, and F.~Many{\`{a}}.
\newblock A modular reduction of regular logic to classical logic.
\newblock In {\em Proc. 31st Int. Symp. on Multiple-Valued Logic}, pages
  221--226, 2001.

\bibitem{BringmannHML15}
K.~Bringmann, D.~Hermelin, M.~Mnich, and E.~J. van Leeuwen.
\newblock Parameterized complexity dichotomy for steiner multicut.
\newblock In {\em Proc. 32nd STACS}, pages 157--170, 2015.
\newblock \href {http://dx.doi.org/10.4230/LIPIcs.STACS.2015.157}
  {\path{doi:10.4230/LIPIcs.STACS.2015.157}}.

\bibitem{CoppersmithV85}
D.~Coppersmith and U.~Vishkin.
\newblock Solving {NP}-hard problems in `almost trees': {V}ertex cover.
\newblock {\em Discrete Applied Mathematics}, 10(1):27--45, 1985.
\newblock \href {http://dx.doi.org/10.1016/0166-218X(85)90057-5}
  {\path{doi:10.1016/0166-218X(85)90057-5}}.

\bibitem{DellM14}
H.~Dell and D.~van Melkebeek.
\newblock Satisfiability allows no nontrivial sparsification unless the
  polynomial-time hierarchy collapses.
\newblock {\em J. {ACM}}, 61(4):23:1--23:27, 2014.
\newblock \href {http://dx.doi.org/10.1145/2629620}
  {\path{doi:10.1145/2629620}}.

\bibitem{DomLS14}
M.~Dom, D.~Lokshtanov, and S.~Saurabh.
\newblock Kernelization lower bounds through colors and {ID}s.
\newblock {\em {ACM} Transactions on Algorithms}, 11(2):13, 2014.
\newblock \href {http://dx.doi.org/10.1145/2650261}
  {\path{doi:10.1145/2650261}}.

\bibitem{DowneyF13}
R.~G. Downey and M.~R. Fellows.
\newblock {\em Fundamentals of Parameterized Complexity}.
\newblock Texts in Computer Science. Springer, 2013.

\bibitem{FellowsJR13}
M.~R. Fellows, B.~M.~P. Jansen, and F.~Rosamond.
\newblock Towards fully multivariate algorithmics: Parameter ecology and the
  deconstruction of computational complexity.
\newblock {\em European J. Combin.}, 34(3):541--566, 2013.
\newblock \href {http://dx.doi.org/10.1016/j.ejc.2012.04.008}
  {\path{doi:10.1016/j.ejc.2012.04.008}}.

\bibitem{Fiala01}
J.~Fiala, T.~Kloks, and J.~J. Kratochv{\'{\i}}l.
\newblock Fixed-parameter complexity of $\lambda$-labelings.
\newblock {\em Discrete Applied Mathematics}, 113(1):59--72, 2001.
\newblock \href {http://dx.doi.org/10.1016/S0166-218X(00)00387-5}
  {\path{doi:10.1016/S0166-218X(00)00387-5}}.

\bibitem{FlumG06}
J.~Flum and M.~Grohe.
\newblock {\em Parameterized Complexity Theory}.
\newblock Springer-Verlag New York, Inc., 2006.

\bibitem{GuoN06}
J.~Guo and R.~Niedermeier.
\newblock Exact algorithms and applications for tree-like weighted set cover.
\newblock {\em Journal of Discrete Algorithms}, 4(4):608--622, 2006.
\newblock \href {http://dx.doi.org/10.1016/j.jda.2005.07.005}
  {\path{doi:10.1016/j.jda.2005.07.005}}.

\bibitem{ImpagliazzoPZ01}
R.~Impagliazzo, R.~Paturi, and F.~Zane.
\newblock Which problems have strongly exponential complexity?
\newblock {\em J. Comput. Syst. Sci.}, 63(4):512--530, 2001.
\newblock \href {http://dx.doi.org/10.1006/jcss.2001.1774}
  {\path{doi:10.1006/jcss.2001.1774}}.

\bibitem{Karp72}
R.~M. Karp.
\newblock {Reducibility Among Combinatorial Problems}.
\newblock In {\em Complexity of Computer Computations}, pages 85--103. Plenum
  Press, 1972.

\bibitem{LuLTLX14}
M.~Lu, T.~Liu, W.~Tong, G.~Lin, and K.~Xu.
\newblock Set cover, set packing and hitting set for tree convex and tree-like
  set systems.
\newblock In {\em Proc. 11th TAMC}, pages 248--258, 2014.
\newblock \href {http://dx.doi.org/10.1007/978-3-319-06089-7_17}
  {\path{doi:10.1007/978-3-319-06089-7_17}}.

\bibitem{Manya00}
F.~Many{\`{a}}.
\newblock The 2-{SAT} problem in signed {CNF}-formulas.
\newblock {\em Multiple-Valued Logic}, 2000.

\bibitem{Niedermeier10}
R.~Niedermeier.
\newblock Reflections on multivariate algorithmics and problem
  parameterization.
\newblock In {\em Proc. 27th STACS}, pages 17--32, 2010.
\newblock \href {http://dx.doi.org/10.4230/LIPIcs.STACS.2010.2495}
  {\path{doi:10.4230/LIPIcs.STACS.2010.2495}}.

\bibitem{Trick87}
M.~A. Trick.
\newblock Induced subtrees of a tree and the set packing problem.
\newblock Technical Report 377, Institute for Mathematics and Its Applications,
  1987.

\bibitem{UhlmannW13}
J.~Uhlmann and M.~Weller.
\newblock Two-layer planarization parameterized by feedback edge set.
\newblock {\em Theor. Comput. Sci.}, 494:99--111, 2013.
\newblock \href {http://dx.doi.org/10.1016/j.tcs.2013.01.029}
  {\path{doi:10.1016/j.tcs.2013.01.029}}.

\bibitem{Wahlstrom07}
M.~Wahlstr{\"o}m.
\newblock {\em Algorithms, Measures and Upper Bounds for Satisfiability and
  Related Problems}.
\newblock PhD thesis, Link\"opings universitet, Sweden, 2007.

\end{thebibliography}

\clearpage

\appendix

\section{Reducing signed 2-SAT to classic 2-SAT} \label{app:twosat}
Given a totally ordered regular signed 2-SAT formula, replace each literal of the form~$x_i \geq j$ by~$\neg (x_i \leq j-1)$, noting that if~$j = 1$ the clause is always satisfied and can be removed instead. The resulting clauses consist of literals~$x_i \geq j$ and~$\neg (x_i \geq j)$ for~$j \in [N]$. To construct a classical 2-SAT formula, we interpret every term of the form~$x_i \geq j$ as a new variable, such that we have a variable~$x_i \geq 1$, another variable~$x_i \geq 2$, and so on. The resulting 2-SAT formula over this new set of variables consists of the clauses resulting from our conversion process above, together with clauses~$(x_i \geq j+1 \Rightarrow x_i \geq j)$ for all~$i \in [n]$ and~$j \in [N-1]$. Observe that such clauses may also be represented as~$(\neg (x_i \geq j+1) \vee (x_i \geq j))$, which shows they are valid 2-clauses. Finally, we add singleton clauses~$(x_i \geq 1)$ for all~$i \in [n]$. We invite the reader to verify that the resulting classical formula on variables~$(x_1 \geq 1), \ldots, (x_1 \geq N), \ldots, (x_n \geq 1), \ldots, (x_n \geq N)$ is classically satisfiable if and only if the signed formula is satisfiable over truth value set~$[N]$.

\end{document}